\documentclass[10pt]{article}

\usepackage[margin=1in]{geometry}
\usepackage{natbib}
\usepackage{parskip}

\usepackage{amssymb}
\usepackage{amsmath}
\usepackage{graphicx}
\usepackage{mathtools}
\usepackage{needspace} 
\usepackage{multicol}
\usepackage{listings}
\usepackage{amsthm}
\usepackage[shortlabels]{enumitem}
\usepackage[italicdiff]{physics}
\usepackage{cancel}
\usepackage[dvipsnames]{xcolor}
\usepackage{verbatim}
\usepackage{comment}
\usepackage{array}
\usepackage{tikz-cd}
\usepackage{import}
\usepackage{booktabs}
\usepackage{comment}
\usepackage{array}
\usepackage[colorlinks, citecolor=blue]{hyperref}
\usepackage{cleveref}
\usepackage{forest}
\usepackage{float}
\usepackage{xspace}
\usepackage{aliascnt}
\usepackage{algorithm2e}
\usepackage{wrapfig}
\usepackage{thm-restate}
\usepackage{adjustbox}
\usepackage{bm}
\usepackage{bbm}
\usepackage{complexity}
\usepackage{makecell}
\usetikzlibrary{calc, positioning, external, arrows.meta}

\usepackage{autonum} %
\makeatletter
\autonum@generatePatchedReferenceCSL{Cref}
\autonum@generatePatchedReferenceCSL{cref}
\makeatother

\renewcommand{\R}{\mathbb R}
\newcommand{\Id}{\op{\mathsf Id}}
\newcommand{\supp}{\op{supp}}

\newcommand{\ep}{\epsilon}

\let\E\relax

\DeclareMathOperator*{\E}{\mathbb E}
\let\co\relax
\DeclareMathOperator*{\co}{co}

\let\ds\displaystyle
\let\op\operatornamewithlimits
\let\eps\epsilon
\let\ip\ev
\let\cite\citep

\newtheorem{theorem}{Theorem}
\numberwithin{theorem}{section}
\newtheorem*{theorem*}{Theorem}

\newtheorem{lemma}[theorem]{Lemma}
\newtheorem{proposition}[theorem]{Proposition}
\newtheorem{corollary}[theorem]{Corollary}

\theoremstyle{definition}

\newcommand{\ind}[1]{\mathbbm{1}\qty{#1}}

\newcommand{\ie}{{\em i.e.}\xspace}
\newcommand{\eg}{{\em e.g.}\xspace}

\newcommand{\nb}[3]{}

\newcommand{\br}[1]{\nb{Brian}{blue}{#1}}

\newcommand{\noah}[1]{\nb{Noah}{purple}{#1}}
\def\va{{\bm{a}}}

\def\ve{{\bm{e}}}

\def\vu{{\bm{u}}}

\def\vx{{\bm{x}}}

\let\vec\bm

\def\cA{{\mathcal{A}}}

\def\cD{{\mathcal{D}}}

\def\cO{{\mathcal{O}}}

\def\cT{{\mathcal{T}}}

\def\cX{{\mathcal{X}}}

\def\tO{{\widetilde{\cO}}}

\usepackage{caption}
\makeatletter
\def\thanks#1{\protected@xdef\@thanks{\@thanks
        \protect\footnotetext{#1}}}
\makeatother

\title{A Lower Bound on Swap Regret in Extensive-Form Games}
\author{
    Constantinos Daskalakis\thanks{Authors alphabetically ordered.} \\
    \texttt{costis@csail.mit.edu} \\
    MIT  \\ Archimedes AI
\and
    Gabriele Farina \\
    \texttt{gfarina@mit.edu} \\
    MIT
\and
    Noah Golowich \\
    \texttt{nzg@mit.edu} \\
    MIT
\and
    Tuomas Sandholm \\
    \texttt{sandholm@cs.cmu.edu} \\
    Carnegie Mellon University\\
    Strategic Machine, Inc. \\
    Strategy Robot, Inc. \\
    Optimized Markets, Inc. 
\and
    Brian Hu Zhang\\
    \texttt{bhzhang@cs.cmu.edu} \\
    Carnegie Mellon University
}
\date{\today}

\begin{document}

\maketitle

\begin{abstract}\noindent
    Recent simultaneous works by \citet{Peng23:Fast} and \citet{Dagan23:External} have demonstrated the existence of a no-swap-regret learning algorithm that can reach $\epsilon$ average swap regret against an adversary in any extensive-form game within $m^{\tO(1/\eps)}$ rounds, where $m$ is the number of nodes in the game tree. However, the question of whether a $\poly(m, 1/\eps)$-round algorithm could exist remained open. In this paper, we show a lower bound that precludes the existence of such an algorithm. In particular, we show that achieving average swap regret $\eps$ against an oblivious adversary in general extensive-form games requires at least $\exp(\Omega\qty(\min\qty{m^{1/14}, \eps^{-1/6}}))$ rounds. 
\end{abstract}

\section{Introduction}

{\em No-regret learning} is a popular framework for modeling situations in which an agent faces an arbitrary, possibly adversarial environment. The agent seeks to minimize its {\em regret}, which is the difference between the utility it has earned and the maximum utility it could have earned by changing its strategy according to some {\em strategy transformation function}. The more strategy transformations that are allowed, the tighter the resulting notion of regret. In sequential (extensive-form) imperfect-information games, especially adversarial games, algorithms based on no-regret learning have been pivotal in leading to superhuman performance in games ranging from poker~\cite{Moravvcik17:DeepStack,Brown18:Superhuman,Brown19:Superhuman} to {\em Diplomacy}~\cite{Meta22:Human}.

In games, there is a well-studied and tight connection between no-regret learning in games and solution concepts involving {\em correlation}. In particular, if all players in a game play according to a no-regret learning algorithm that minimizes a certain notion of regret, their average strategy profile will converge to a notion of correlated equilibrium that corresponds to the class of strategy transformations for that regret notion. Notable correlated solution concepts, and their corresponding sets of deviations (notions of regret), can be found in \Cref{tab:eqconcepts}. 

\begin{table}[htp]
\centering
\scalebox{.86}{%
    \begin{tabular}{@{}lllm{4.4cm}@{}}
        \bf Set of deviations & \bf Equilibrium concept & \multicolumn{2}{l}{\bf Best known no-regret guarantee} \\
        \toprule
        Constant functions (``External regret'') & Normal-form coarse correlated (NFCCE) & $m/\eps^2$ &\cite{Farina22:Simple} \\ \midrule
        \makecell[l]{``Trigger'' functions \\ Linear functions (``Linear swap regret'')}  & \makecell[l]{Extensive-form correlated (EFCE) \\ Linear-swap correlated (LCE)} & $m^2/\eps^2$ & \cite{Farina22:Simple}\newline \cite{Farina23:Polynomial} \\ \midrule
        All functions (``Swap regret'') & Normal-form correlated (NFCE) &$m^{\tO(1/\eps)}$ & \cite{Peng23:Fast} \newline \cite{Dagan23:External} \\
        \bottomrule
    \end{tabular}
}
\caption{Some examples of notions of regret and corresponding notions of correlated equilibrium for extensive-form games, in increasing order of tightness. The ``best no-regret guarantee'' is the minimum number of iterations $T$ after which an agent can guarantee (average) regret $\eps$ in that notion. $m$ is the size of the game (number of nodes). All algorithms listed have guaranteed $\poly(m, 1/\eps)$ per-iteration time complexity.}\label{tab:eqconcepts}
\end{table}

The tightest notion that can be defined in this manner is known as the {\em normal-form correlated equilibrium} (NFCE)~\cite{Aumann74:Subjectivity}, which corresponds to measuring regret against the set of {\em all possible strategy transformations}---known as {\em swap regret}~\cite{Blum07:External}. It has long been believed that efficiently computing an NFCE or minimizing swap regret is impossible. In fact, the believed hardness of computing NFCE has motivated the development of many more relaxed notions of equilibrium with nicer computational properties, including the extensive-form correlated equilibrium~\cite{Stengel08:Extensive}, linear-swap correlated equilibrium~\cite{Farina23:Polynomial}, and behavioral correlated equilibrium~\cite{Morrill20:Hindsight,Morrill21:Efficient,Zhang24:Outcome}. \noah{does behavioral CE have nicer properties? IIRC we don't know how to compute it}\br{yes :) see above cited paper of mine} Our focus in this paper, however, will be on NFCE itself.

Little was known about the feasibility of swap regret minimization (and, similarly, about computing NFCE) until a recent simultaneous work due to \citet{Peng23:Fast} and \citet{Dagan23:External} which implies that, for extensive-form games with $m$ nodes, there is a no-regret algorithm that achieves (average) swap regret $\eps$ after $m^{\tO(1/\eps)}$ rounds. For extensive-form games, this was the first time that a polynomial-round swap regret minimization algorithm (even for constant $\eps$) had been achieved, and resulted in the first PTAS for computing an NFCE of an extensive-form game. Both papers provide lower bounds, but these lower bounds are only for normal-form (\ie, single-step) games, and do not preclude the existence of efficient swap-regret minimization algorithms for extensive-form games. Thus, there is an exponential gap (in terms of the dependence on $\eps$) between the best regret bounds for swap regret and the best regret bounds for weaker notions of regret. 

In this paper, we show that this exponential gap cannot be closed. In particular, we show the following lower bound on swap regret in extensive-form games.
\begin{theorem}[Main theorem, informal]
    There is no swap regret minimization algorithm for general extensive-form games that requires fewer than $\exp(\Omega\qty(\min\qty{m^{1/14}, \eps^{-1/6}}))$ rounds. 
\end{theorem}
Our result implies an exponential gap between swap regret and weaker notions of regret in the realm of extensive-form games. In particular, our result precludes the existence of $\poly(m, 1/\eps)$-time swap regret minimization algorithms. 

\section{Preliminaries}

\paragraph{No-regret learning and $\Phi$-regret.}
In {\em no-regret learning}, a learner with access to a {\em strategy set} $\cX \subset \R^m$ faces an adversarial environment across many rounds. In each round $t = 1, \dots, T$, the learner outputs a distribution $\pi^t \in \Delta(\cX)$ and then the environment outputs a linear \noah{changed from ``affine''} {\em utility function} $u^t : \cX \to [-1, 1]$ which we will henceforth write for convenience as a vector $\vu^t \in \R^m$. The utility function $u^t$ may depend on all the past distributions $\pi^1, \dots, \pi^t$ selected by the learner. The learner then gets utility $\E_{\vx \sim \pi^t}\ip{\vu^t, \vx}$. The environment (i.e., adversary) is said to be \emph{oblivious} if  its choices of utility vectors $\vu^t$ do not depend on the learner's chosen distributions $\pi^s$. 

In this paper, the notion of interest is {\em swap regret}. Intuitively, a learner has low swap regret if it could not have improved its utility by transforming its strategy according to any function $\phi : \cX \to \cX$. More formally, define
\begin{align}
    V(\phi) := \frac{1}{T} \sum_{t=1}^T \E_{\vx \sim \pi^t} \ip{\vu^t, \phi(\vx)}.
\end{align}
Thus, in particular, the total utility experienced by the learner is $V(\Id)$, where $\Id : \cX \to \cX$ is the identity function. After $T$ rounds, the {\em (average) swap-regret} is 
\begin{align}
    \textsc{SwapRegret}(T) := \max_{\phi : \cX \to \cX} V(\phi) - V(\Id) = \max_{\phi : \cX \to \cX} \frac{1}{T} \sum_{t=1}^T  \E_{\vx \sim \pi^t} \ip{\vu^t, \phi(\vx) - \vx}.
\end{align}
The learner's goal is then to achieve small swap regret after a small number of rounds: for example, one may hope to achieve swap regret $\eps$ after $T = \poly(m, 1/\eps)$ rounds.

Other notions of regret, such as those mentioned in the introduction, can be defined by restricting the set of deviations to a set $\Phi \subset \cX^\cX$. Since for this paper we are only interested in swap regret, we will not explore this connection further.

\paragraph{Tree-form decision problems and extensive-form games.}
Tree-form decision problems describe {\em sequential} interactions between the player and the environment. In a tree-form decision problem, there is a rooted tree $\cT$ of {\em nodes}. There are two types of nodes: {\em decision points}, at which the player selects an action, and {\em observation points}, at which the environment selects an observation. At each decision point $j$, actions are identified with outgoing edges, and we use $A_j$ to denote this set. Leaves of $\cT$ are called {\em terminal nodes}. We will use $m$ to denote the number of terminal nodes. 

A {\em pure strategy} is a choice of one action at each decision point. The {\em tree-form} or {\em realization-form representation} of the pure strategy is the vector $\vx \in \{0, 1\}^m$ for which $\vx[z] = 1$ if and only if the player plays {\em all} actions on the path from the root node to terminal node $z$. The set of tree-form pure strategies will be denoted $\cX \subset \{0, 1\}^m$. 
Different pure strategies may have the same tree-form representation, but for our purposes we will only require the tree-form representation of strategies, and therefore we will not distinguish between strategies with the same tree-form representation. A {\em mixed strategy} $\pi \in \Delta(\cX)$ is a distribution over pure strategies.

Tree-form decision problems naturally model the decision problems faced by players in an {\em extensive-form game}. For our purposes, a (perfect-recall) extensive-form game with $n$ players is defined by $n$ tree-form strategy sets $\cX_1, \dots, \cX_n$, and one $n$-linear {\em utility function} $u_i : \cX_1 \times \dots \times \cX_n \to [-1, 1]$ for each player $i \in [n]$, which defines the utility of player $i$ when each player $j \in [n]$ plays strategy $\vx_j \in \cX_j$. 

Our solution concept of interest is the  {\em (normal-form) correlated equilibrium} (NFCE)~\cite{Aumann74:Subjectivity}.  An {\em $\eps$-NFCE} is a distribution $\pi \in \Delta(\cX_1 \times \dots \times \cX_n)$ with the property that no player can profit by applying any function $\phi : \cX_i \to \cX_i$ to their strategy. That is, $\pi$ is an $\eps$-NFCE if
\begin{align}
    \E_{(\vx_1, \dots, \vx_n) \sim \pi } \qty[u_i(\phi(\vx_i), \vx_{-i}) - u_i(\vx_i, \vx_{-i})] \le \eps
\end{align}
for all players $i$ and functions $\phi : \cX_i \to \cX_i$.

Suppose we have an $n$-player game played repeatedly over $T$ rounds. For each round $t \in [T]$, let $\pi^t_i \in \Delta(\cX_i)$ be the mixed strategy played by player $i$. 
\begin{proposition}
Suppose that each player plays according to a no-swap-regret learning algorithm using utility maps $u^t_i : \cX_i \to [-1, 1]$ given by $$u^t_i(\vx_i) = \E_{\vx_{-i} \sim \pi^t_{-i}} u_i(\vx_i, \vx_{-i}^t).$$
Let $\pi^t \in \Delta(\cX_1 \times \dots \times \cX_n)$ be the product distribution whose marginal on $\cX_i$ is $\pi^t_i$. 
If the swap-regret of every player $i$ is bounded by $\eps$, then the average profile $\pi := \frac{1}{T} \sum_{i=1}^T \pi^t$ is an $\eps$-NFCE.
\end{proposition}
The proof follows immediately by comparing the definition of swap regret and the definition of NFCE.

\section{Previous Swap Regret Minimization Algorithms}
We now review known results about no-swap-regret learning algorithms.

\subsection{Normal-form games} In a {\em normal-form game}, every player's decision problem consists of a single decision point with $m$ actions, that is, $\cX = \{ \ve_1, \dots, \ve_m\} \subset \{0, 1\}^m$ where $\ve_k$ is the $k$th standard basis vector in $\R^m$. We will abuse language slightly and refer to $\cX$ as the $m$-simplex, even though it is actually the convex hull of $\cX$ that is the $m$-simplex. \citet{Blum07:External} showed that efficient algorithms exist for minimizing swap regret over the simplex.
\begin{theorem}[\citealp{Blum07:External}]
    There exists a no-regret learning algorithm for simplices that achieves average swap regret $\eps$ within $T = \tO(m/\eps^2)$ rounds.
\end{theorem}
One may wonder whether this is optimal, \eg, whether it is possible to achieve a logarithmic dependence on $m$. Recent simultaneous work by \citet{Dagan23:External} and \citet{Peng23:Fast} has essentially completely answered this question for normal-form games.
\begin{theorem}[\citealp{Dagan23:External,Peng23:Fast}, upper bound]\label{th:ub}
    There exists a no-regret learning algorithm for simplices that achieves average swap regret $\eps$ within $T = (\log m)^{\tO(1/\eps)}$ rounds.
\end{theorem}
Both papers also provided (nearly-)matching lower bounds. Here we state a particularly simple-to-state lower bound proven by \citet{Dagan23:External}.
\begin{theorem}[Theorem 4.1 of \citealp{Dagan23:External}, lower bound]\label{th:lb}
  Let $T < m/4$. Then, in the there exists an oblivious %
  adversary such that the swap regret of any learner for the $m$-simplex is $\Omega(\log^{-5}T)$.
\end{theorem}

\subsection{Extensive-form games and tree-form strategy sets} For more general extensive-form games, the picture is less clear. For an upper bound, one can consider a tree-form decision problem with $M$ pure strategies (\ie, $|\cX| = M$) as simply an ``easier version'' of a normal-form decision problem where each pure strategy is treated as a different action, \ie, where the strategy set is the $M$-simplex. \Cref{th:ub} therefore implies a similar bound on swap regret for tree-form decision problems.\footnote{As stated, the bound is only information-theoretic. However, the information-theoretic bound is implementable by an efficient (\ie, $\poly(m, 1/\eps)$-time-per-iteration) algorithm, which is described by \citet{Dagan23:External} and \citet{Peng23:Fast}, and is beyond the scope of this paper.}
\begin{corollary}[\citealp{Dagan23:External,Peng23:Fast}, tree-form upper bound]\label{th:efg-upper}
Let $\cX \subset \{0, 1\}^m$ be a tree-form strategy set. 
    There exists a no-regret learning algorithm for tree-form decision problems that achieves swap regret $\eps$ after $T = (\log |\cX|)^{\tO(1/\eps)} \le m^{\tO(1/\eps)}$ rounds.
\end{corollary}

Showing a matching lower bound for extensive form, however, remained open. The main difficulty is that the adversary is restricted to {\em linear} utility functions $u^t : \cX  \to \R$; the adversary in \Cref{th:lb} does not use linear utility functions when the extensive-form game is interpreted as a normal-form game over $M$ actions as described above. The purpose of this paper is to close this discrepancy, by showing a lower bound that almost matches \Cref{th:efg-upper}. 

\section{Main Result}

Our main result is the following.

\begin{theorem}[Main theorem]\label{th:efg-lower}
    There exist arbitrarily large tree-from strategy sets $\cX \subset \{0, 1\}^m$ with the following property. Let $\eps > 0$ and suppose $T \le \exp(\Omega\qty(\min\qty{m^{1/14}, \eps^{-1/6}}))$. Then there exists an oblivious adversary running for $T$ iterations against which no learner can achieve expected swap regret better than $\eps$.
\end{theorem}

Intuitively, the proof of \Cref{th:efg-lower} works by finding an ``embedding'' of the adversary of \Cref{th:lb} into a tree-form decision problem such that the utility functions $u^t$ do remain linear. This works by choosing random vectors in $\{-1, 1\}^n$ (for some appropriately chosen dimension $n$) to simulate the ``actions'' in the (exponentially large) normal-form decision problem, and then exploiting the concentration property that an exponentially large number of such vectors $\{\va_i\}_{i=1}^{M}$ can be chosen such that $\ip{\va_i, \va_j} \approx 0$ for all $i \ne j$.

Like \Cref{th:lb}, our lower bound is {\em information-theoretic}: it does not rely on computational hardness results, and thus applies to {\em any} no-regret learning algorithm no matter how much computation it might perform.

Before proving \Cref{th:efg-lower}, we first state a more detailed version of the normal-form lower bound (\Cref{th:lb}).\footnote{The discussion in \citealt{Dagan23:External} on pages 37--38 of  specifies the adversary which satisfies the properties listed in \Cref{th:normalform}.} This restatement changes the notation so as to avoid mixing the notation between tree-form and normal-form decision problems, and extracts some useful properties of the adversary. In particular, a key property of the adversary that we exploit (specified in \Cref{it:adv-unit} in the below theorem) is that, with probability 1, each of the vectors $\vu^t$ that it chooses is in fact a unit vector $\ve_i$. (In general, in the normal-form game setting, the vectors $\vu^t$ may have coordinates in $[-1,1]$.)

\begin{theorem}[\citealt{Dagan23:External}, expanded version of \Cref{th:lb}]\label{th:normalform}
    Let $\cA = \{ \ve_1, \ldots, \ve_M \}$ be the $M$-simplex, and let $T < M/4$. Then there exists an adversary on $\cA$ with the following properties:
    \begin{enumerate}
        \item  \label{it:adv-unit} The adversary selects a sequence $(\vu^1, \dots, \vu^T) \sim \cD$ from some distribution $\cD \in \Delta(\cA^T)$, and then outputs utility vector $\vu^t$ at time $t$ regardless of the sequence of distributions played by the learner.
        \item There exists an action $\va^* \in \cA$ that is never used by the adversary.\footnote{This can always be assumed WLOG.} \noah{maybe note saying this can be added wlog? I don't think our adversary technically has it}
        \item There exists a partition $\cA = \cA_1 \sqcup \dots \sqcup \cA_d$ where $d \le \cO(\log T)$ with the following property. 
        Within each set $\cA_i$, number the actions $\cA_i = \{ \va_{i1}, \dots, \va_{iM_i} \}$. For any sequence $(\vu^1, \dots, \vu^T)\in \supp \cD$, the adversary plays actions in $\cA_i$ only in increasing order. That is, if $\vu^t = \va_{ij}$ and $\vu^{t'} = \va_{ij'}$ and $t \le t'$, then $j \le j'$. 
        \item The swap regret of any learner against this adversary is\footnote{The reason that the $-5$ in \Cref{th:lb} has been changed to a $-6$ here is that the adversary is now constrained to pick a sequence of {\em actions}, \ie, $\ell_1$-bounded losses, instead of $\ell_\infty$-bounded losses. See \cite{Dagan23:External}, Theorems 1.7 and 4.1} $\Omega(d^{-6}) = \Omega(\log^{-6} T)$. %
        \label{prop:increasing}
    \end{enumerate}
\end{theorem}

We now prove \Cref{th:efg-lower} via a reduction from \Cref{th:normalform}. In particular, we fix $M \in \mathbb{N}$ as in \Cref{th:normalform}, and let $\cA = \{ \ve_1, \ldots, \ve_M \}$ denote the action set for the lower bound in \Cref{th:normalform}. We will often use the decomposition $\cA = \cA_1 \sqcup \cdots \sqcup \cA_d$. 

\paragraph{Extensive-form game instances.} Consider the following family of tree-from strategy sets, parameterized by natural numbers $d$ and $n$. First the learner picks an index $i \in [d]$. Then the environment picks $j \in [n]$, and finally the learner picks a binary action. This family of decision problems is depicted in \Cref{fig:dp}.

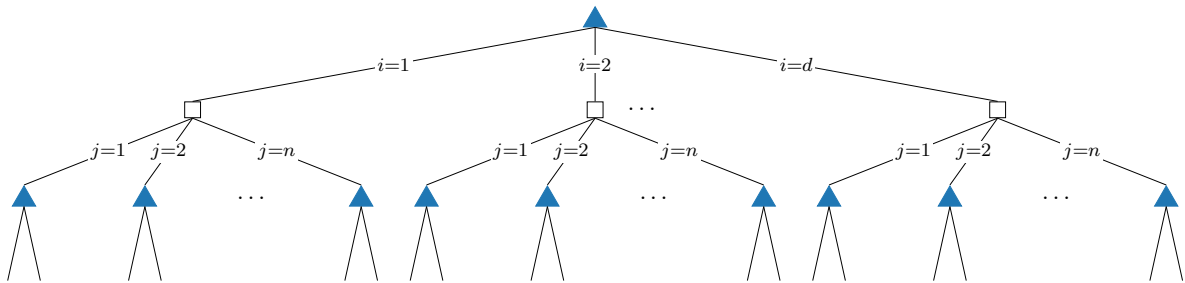
\begin{figure}[h]
\definecolor{p1color}{RGB}{31,119,180}
\newcommand{\pone}{{\ensuremath{\color{p1color}\blacktriangle}}\xspace}
\tikzset{
  every path/.style={-},
}
\forestset{
        default preamble={for tree={
        parent anchor=south, child anchor=north
}},
  p1/.style={
      regular polygon,
      regular polygon sides=3,
      inner sep=2pt, fill=p1color, draw=p1color,
  },
  nat/.style={draw},
  terminal/.style={draw=none, inner sep=2pt},
   el/.style n args={1}{edge label={node[midway, fill=white, inner sep=1pt, draw=none, font=\footnotesize] {#1}}},
}
\centering
\scalebox{0.9}{
\begin{forest}
for tree={l+=1em},
    [,p1
        [,nat,el={$i{=}1$},s sep=3em
            [,p1,el={$j{=}1$} [] []]
            [,p1,el={$j{=}2$} [] []]
            [$\dots$, draw=none, no edge]
            [,p1,el={$j{=}n$} [] []]
        ]
        [,nat,el={$i{=}2$},s sep=3em
            [,p1,el={$j{=}1$} [] []]
            [,p1,el={$j{=}2$} [] []]
            [$\dots$, draw=none, no edge]
            [,p1,el={$j{=}n$} [] []]
        ]
        [$\dots$, draw=none, no edge]
        [,nat,el={$i{=}d$},s sep=3em
            [,p1,el={$j{=}1$} [] []]
            [,p1,el={$j{=}2$} [] []]
            [$\dots$, draw=none, no edge]
            [,p1,el={$j{=}n$} [] []]
        ]
    ]
\end{forest}
}
    \caption{A depiction of the class of tree-form decision problems used in the proof of \Cref{th:efg-lower}. Triangles (\pone) are decision points and boxes ($\square$) are observation points. }\label{fig:dp}
\end{figure}

A pure strategy is identified (up to linear transformations) by a vector $\vx \in \R^{d\times n}$ where, for some $i \in [d]$, $\vx[i, \cdot] \in \{-\frac{1}{\sqrt{n}}, \frac{1}{\sqrt{n}} \}^n$ and $\vx[i', \cdot]=\vec 0$ if $i' \ne i$ (\ie, $\vx$ interpreted as a matrix with exactly one nonzero row). For convenience, we will use $\cX_i \subset \cX$ to denote the set of (pure) strategies where the learner plays $i$ at the root. Let $C$ be an absolute constant large enough to make the asymptotic bounds in \Cref{th:normalform} true.

\paragraph{The adversarial environment.}
The adversary used to prove \Cref{th:efg-lower} works as follows. First, for each $i \in [d]$, it populates $\cA_i$ with $M_i$ uniformly randomly chosen strategies in $\cX_i$. %
Formally, we let $\psi : \cA \to \cX$ denote the (random) mapping which associates each action in $\cA_i \subset \cA$ with the corresponding action in $\cX_i$: the image of $\cA_i$ under $\psi$ consists of actions we denote by $\tilde\va_{i1}, \ldots, \tilde\va_{iM_i} \in \cX_i$, which are chosen independently and uniformly in $\cX_i$. The adversary in \Cref{th:normalform} produces a random sequence $\vu^1, \ldots, \vu^T \in \cA$; we consider the adversary which draws a sequence $\vu^t$ from that distribution and outputs the sequence consisting of $\tilde \vu^t := \psi(\vu^t)$ for $t \in [T]$. %
\noah{I changed notation a bit here since $\cA$ and $\cX$ were being conflated}

\paragraph{Analysis.}
Let $\eps > 0$ be a parameter to be selected later. We start with a simple concetration bound.
\begin{lemma}\label{lem:concentration}
    Let $\delta = M^2 e^{-n\eps^2/2}$. With probability $1 - \delta$, for all $\va, \va ' \in \cA$, we have $\abs{\ip{\tilde{\vec a}, \tilde{\vec a}'} - \ind{\va = \va'}} \le \eps$.
\end{lemma}
\begin{proof}
    If $\va = \va '$ then the claim holds trivially because then $\tilde{\va} = \tilde{\va}'$, and they are both unit vectors. For a fixed $\vec a \ne \vec a' \in \cA$, the claim holds with probability $2e^{-n\eps^2/2}$ by Hoeffding's inequality. The lemma then follows by union bounding over the $\binom{M}{2} \le M^2/2$ pairs.
\end{proof}

We will claim that, for any learner against this adversary, there exists a learner against the adversary of \Cref{th:normalform} that achieves a similar swap regret---and thus the swap regret of the former learner must be large. First, we will construct the latter learner.

Let $\pi^1, \dots, \pi^T \in \Delta(\cX)$ be the sequence of distributions played by the learner. Note that $\pi^t$ can depend on the utilities $\vu^{1:t-1} \in \cA$ that are played by the adversary. Consider the sequence $\bar\pi^1, \dots, \bar\pi^T \in \Delta(\cA)$, where $\bar\pi^t$ is the distribution that samples $\vx \sim \pi^t$ and plays according to $p_\vx \in \Delta(\cA)$, defined as follows. Let $\vx \in \cX_i$ be any strategy. There are two cases.
\begin{enumerate}
    \item 
    $\ip{\vx,\tilde \va_{ij}} \le \eps$ for every $i \in [d], j \in [M_i]$. Then define $p_\vx=\va^*$ deterministically (i.e., $p_\vx$ is the distribution which puts all of its mass on $\va^*$). 
    \item $\ip{\vx, \tilde\va_{ij}} > \eps$ for some $i \in [d], j \in [M_i]$. Let $j$ be the {\em largest} such index, let $\beta = \ip{\vx, \tilde\va_{ij}}$, and define $p_\vx$ as the distribution that is $\va^*$ with probability $1-\beta$ and $\va_{ij}$ with probability $\beta$. Note that $\beta \in [0,1]$ since $\vx, \tilde\va_{ij} \in \{ -\frac{1}{\sqrt n}, \frac{1}{\sqrt n} \}^n$ and in this case we have assumed that $\langle \vx, \tilde\va_{ij} \rangle > \eps > 0$. 
\end{enumerate}
A critical property for us will be that the learner cannot ``guess in advance'' what future unobserved $\tilde\va_{ij}$s will be, since these are sampled uniformly at random. That is, in Case 2, $\vx$ can only be played with large probability once the adversary has played $\tilde\va_{ij}$. 

To be more formal, we first define some notation. For every $i \in [d], j \in [M_i]$, let $t_{ij}$ be the first iteration on which the  adversary plays $\tilde\va_{ij}$ (or $t_{ij} = T$ if this never happens). For $\vx \in \cX_i$, if $\vx$ is in Case 1 above then define $t_\vx = 0$, and otherwise define $t_\vx = t_{ij}$, where $j$ is as in Case 2.

There are two properties that we will critically need to use about $t_\vx$. The first states that the learner cannot place large mass on $\vx$ until after $t_\vx$, because doing so would require the learner to guess a vector heavily correlated with $\tilde\va_{ij}$ before the learner observes $\tilde\va_{ij}$.

\begin{lemma}
  \label{lem:delta-bound} $\ds \E\frac{1}{T}\sum_{\vx \in \cX} \sum_{t =1}^{t_\vx} \pi^t(\vx) \le \delta$.
\end{lemma}
\begin{proof}
    Since the learner has not yet observed $\tilde\va_{ij}$ at time $t_{ij}$, its prior strategy sequence $\pi^{1:t_{ij}}(\vx)$ must be independent of $\tilde\va_{ij}$. Moreover, if $t \le t_\vx$ then there must exist some $j$ with $t_{ij} \ge t$ and $\ip{\vx, \tilde\va_{ij}} \ge \eps$---namely, the $j$ defining Case 2. Thus we have:
    \begin{align}
        \E \frac{1}{T} \sum_{\vx \in \cX} \sum_{t=1}^{t_\vx}  \pi^t(\vx)
        &\le  \E \frac{1}{T} \sum_{i=1}^d \sum_{\vx \in \cX_i} \sum_{t=1}^T   \pi^t(\vx) \sum_{j : t_{ij} \ge t} \ind{\ip{\vx, \tilde\va_{ij}} \ge \eps}
        \\&=  \underbrace{\frac{1}{T} \sum_{i=1}^d \sum_{\vx \in \cX_i} \sum_{j=1}^{M_i} \E\qty[  \sum_{t \le t_{ij}} \pi^t(\vx) ]}_{\le M} \underbrace{ \E \qty[ \ind{\ip{\vx,\tilde \va_{ij}} \ge \eps} ] }_{\le e^{-n\eps^2/2}} \le \delta,
    \end{align}
    where in the last line we use the fact that $\tilde\va_{ij}$ is independent of $\pi^{1:t_{ij}}(\vx)$ and then Hoeffding's inequality. Moreover, we have used in the final inequality that for each $i \in [d]$ and $j \leq M_i \leq M$, we have $\frac 1T \sum_{\vx \in \cX}\sum_{t=1}^T \E[\pi^t(\vx)] \leq 1$. 
\end{proof}
The second property is that, for  $t > t_\vx$, utilities of $\vx$ under $\tilde\vu^t$ are approximately the same as those of $p_\vx$ under the utilies $\vu^t$ of \Cref{th:normalform}.
\begin{lemma}
  \label{lem:epsilon-bound}
    For $t > t_\vx$, we have $\ds \ip{\vx, \tilde\vu^t} \le p_\vx(\vu^t) + \eps$.
\end{lemma}
\begin{proof}
Let $\vx \in \cX_i$. There are two cases. In the first case, we suppose that $\ip{\vx, \tilde\va_{ij}} \le \eps$ for every $i \in [d], j \in [M_i]$. Then for every $t$, we have $\vu^t \notin \supp p_\vx = \{ \va^* \}$ (because the adversary of \Cref{th:normalform} never plays $\va^*$), and $\ip{\vx, \tilde\vu^t} \le \eps$ by definition, so we are done.

Otherwise, let $j$ be the largest index for which $\ip{\vx, \tilde\va_{ij}} > \eps$. Then $t_\vx = t_{ij}$ by definition, and since $t > t_{ij}$, by Property~\ref{prop:increasing}, for time steps following $t_{ij}$ the adversary of \Cref{th:normalform} is no longer allowed to play $\va_{ij'}$ for $j' < j$. Thus, either $\vu^t = \va_{ij}$, or else $\vu^t \not \in \{ \va_{i1}, \ldots, \va_{ij} \}$. Since $j$ is defined to be the largest index for which $\langle \vx, \tilde\va_{ij} \rangle > \ep$, in the latter case we must have $p_\vx$ puts all its mass on $\va^* \neq \vu^t$, meaning that $\langle \vx, \tilde\vu^t \rangle \leq \eps = p_\vx(\vu^t) + \eps$. In the former case, we have $\langle \vx, \tilde\vu^t \rangle = \beta = p_\vx(\va_{ij})= p_\vx(\vu^t)$. 
\end{proof}

For the rest of this proof we will use $\bar V(\phi)$ to denote the utilities experienced by $\bar\pi^t$ under the utilities $\vu^t$ in \Cref{th:normalform}. That is,
\begin{align}
    \bar V(\phi) = \frac1T \sum_{t=1}^T \sum_{\va \in \cA} \bar\pi^t(\va) \ind{\phi(\va) = \vu^t} = \frac1T \sum_{t=1}^T \sum_{\vx \in \cX} \pi^t(\vx) \Pr_{\va \sim p_\vx}[\phi(\va) = \vu^t]
\end{align}

By \Cref{th:normalform}, there exists a function $\bar\phi : \cA \to \cA$ such that\footnote{Technically $\phi$ is a random variable dependent on $\vu^1, \dots, \vu^T$.} $\E[\bar V(\bar\phi) - \bar V(\Id)] \ge 1/Cd^6$. \noah{added defn of $\phi$} 
We define a deviation $\phi : \cX \to \co \cX$ by setting\footnote{Note that if a profitable deviation $\cX \to \co \cX$ exists, then by linearity, so must a profitable deviation $\cX \to \cX$.} $\phi(\vx) := \E_{\va \sim p_\vx} \psi(\bar\phi(\va)).$

It suffices to show that $\E[V(\phi) - V(\Id)]$ is large. To do this, we will show that, in expectation and up to small errors, $V(\Id) \le \bar V(\Id)$ and $V(\phi) \ge \bar V(\bar\phi)$.

For the first approximation, we have
\begin{align}
    V(\Id) &= \frac{1}{T} \sum_{\vx \in \cX} \sum_{t=1}^T \pi^t(\vx) \ip{\vx, \tilde \vu^t}
    \\&\le \frac{1}{T} \sum_{\vx \in \cX} \sum_{t > t_\vx} \pi^t(\vx) \ip{\vx, \tilde \vu^t} + \delta
    \\&\le \frac{1}{T} \sum_{\vx \in \cX} \sum_{t > t_\vx} \pi^t(\vx) p_\vx(\vu^t) + \eps + \delta
    \\&\le \frac{1}{T} \sum_{\vx \in \cX} \sum_{t=1}^T \pi^t(\vx) p_\vx(\vu^t) + \eps + 2 \delta = \bar V(\Id) + \eps + 2 \delta,\label{eq:vid}
\end{align}
where the first and third inequalities use \Cref{lem:delta-bound}, and the second inequality uses \Cref{lem:epsilon-bound}. 
For the second, conditional on the event in \Cref{lem:concentration}, we have \noah{changed some things below (e.g. since $\phi$ wasn't defined before)}
\begin{align}
V(\phi) &= \frac{1}{T} \sum_{\vx \in \cX} \sum_{t=1}^T \pi^t(\vx) \ip{\phi(\vx), \tilde\vu^t}
\\&\ge \frac{1}{T} \sum_{\vx \in \cX} \sum_{t > t_\vx} \pi^t(\vx) \ip{ \phi(\vx), \tilde\vu^t} - \delta
\\&\ge \frac{1}{T} \sum_{\vx \in \cX} \sum_{t > t_\vx} \pi^t(\vx) \Pr_{\va \sim p_\vx}[\bar\phi(\va) = \vu^t] - \eps - \delta
\\&\ge  \sum_{\vx \in \cX} \frac{1}{T} \sum_{t=1}^T \pi^t(\vx) \Pr_{\va \sim p_\vx} [\bar \phi(\va) = \vu^t]- \eps - 2\delta = \bar V(\bar\phi) - \eps - 2\delta,
\end{align}
where the first and third inequalities use \Cref{lem:delta-bound}. To establish the second inequality above, we note that, 
\begin{align}
    \ip{ \phi(\vx), \tilde\vu^t} = \E_{\va \sim p_\vx} \ip{\psi(\bar\phi(\va)), \tilde\vu^t} \ge \Pr_{\va \sim p_\vx}[\bar\phi(\va) = \vu^t] - \eps
\end{align}
by \Cref{lem:concentration}, since $\bar\phi(\va), \tilde\vu^t \in \cA$. Thus, accounting for the probability $\delta$ in which \Cref{lem:concentration} fails, we have
\begin{align}
    \E[V(\phi) - \bar V(\bar\phi)] = \underbrace{\E[V(\phi) - \bar V(\bar\phi)|F]}_{\ge -\eps - 2\delta} \cdot \underbrace{\Pr[F]}_{\le 1} + \underbrace{ \E[V(\phi) - \bar V(\bar\phi) | \neg F] }_{\ge -1} \cdot \underbrace{\Pr[\neg F]}_{\le \delta} \ge - \eps - 3\delta\label{eq:vphi}
\end{align}
where $F$ is the event in \Cref{lem:concentration}.
Combining \Cref{eq:vid,eq:vphi},
\begin{align}
  \E[V(\phi) - V(\Id)] \ge \E\qty[\bar V(\bar\phi) - \bar V(\Id)] - 2 \eps - 5 \delta \ge \frac{1}{Cd^6} - 2 \eps - 5 \delta \ge \frac{1}{4Cd^6} = \eps\label{eq:vphi-vid}
\end{align}
by setting the parameters 
\begin{align}
    \eps = \frac{1}{4Cd^6}, \qq{and}
    n = \frac{2\log 20CM^2 d^6}{\eps^2} \qq{so that}  5\delta = 5M^2e^{-n\eps^2/2} \le \frac{1}{4Cd^6}.
\end{align}

\br{this is already captured by $\delta$.}\noah{I believe I was thinking of something slightly different here, namely second-to-last paragraph on page 7? (Perhaps one could combine that failure event with $\delta$, though the way things are written that isn't really being done...)}
The resulting tree-form decision problem hence has dimension $m = d \cdot n = \cO(\log^{14} M)$, and since $d = \Theta(\eps^{-{1/6}}) \le \cO(\log T)$ we have that the swap regret is at least $\eps$ for all $T < \min\qty{M/4, \exp(\Omega(\eps^{-1/6}))}$, where $M = \exp(\Omega(m^{1/14}))$, as desired.

\section{Conclusion and Future Research}
By extending a recent lower bound for normal-form games, we have established a lower bound that precludes the existence of fully polynomial-time (\ie, $\poly(m, 1/\eps)$) algorithms for swap regret minimization in tree-form decision problems. 

Our result leaves open several natural questions for future research.
\begin{enumerate}
    \item Our counterexample applies only for extensive-form games with a particular structure. Is swap regret minimization also information-theoretically impossible for simpler structures, such as {\em single-stage} Bayesian games, in which the strategy set is a product of simplices?
    
    \item Are there uncoupled learning dynamics that yield $\poly(m, 1/\eps)$-time convergence to NFCE when applied in games? Our result does not preclude this possibility, since the behavior of the adversary in \Cref{th:efg-lower} is likely not the behavior of any learning agent in a game. 
        
    \item What is the complexity of {\em computing} one NFCE in an extensive-form game (by any method, not just limited to independent learning dynamics)? This is a problem that was stated as early as \citet{Papadimitriou05:Computing} and \citet{Stengel08:Extensive} and remains open.
\end{enumerate}
    Most correlated notions of equilibrium that {\em are known} to be efficiently computable have corresponding efficient no-regret learning algorithms with convergence guarantee of the form $\poly(m, 1/\eps)$. For example, the notions of correlated equilibrium mentioned in the introduction (up to and including linear correlated equilibria)  admit both efficient no-regret algorithms~\cite{Zhang23:Mediator} and efficient algorithms for exact computation~\cite{Farina24:Polynomial}. We thus believe that our main result, which precludes efficient no-regret learning for swap regret, is evidence against the existence of an efficient algorithm (learning dynamics or otherwise) for computing an NFCE in an extensive-form game. Proving or disproving this claim remains an important open question.

\section*{Acknowledgements}

This material is based on work supported by the Vannevar Bush Faculty
Fellowship ONR N00014-23-1-2876, National Science Foundation grants
RI-2312342 and RI-1901403, ARO award W911NF2210266, and NIH award
A240108S001. Brian Hu Zhang is also supported by the CMU Computer Science Department Hans Berliner PhD Student Fellowship. Constantinos Daskalakis is supported by NSF Awards CCF-1901292, DMS-2022448, and DMS2134108, a
Simons Investigator Award, and the Simons Collaboration on the Theory of Algorithmic Fairness. Noah Golowich is supported by a Fannie \& John Hertz Foundation Fellowship and an NSF Graduate Fellowship.

\bibliographystyle{plainnat}
\bibliography{dairefs}

\end{document}